\documentclass[copyright,creativecommons]{eptcs}
\providecommand{\event}{GandALF 2013} 
\usepackage{breakurl}             

\usepackage[T1]{fontenc}

\usepackage{latexsym}
\usepackage{amssymb}
\usepackage{amsthm}
\usepackage{amsmath}

\usepackage{amstext}
\usepackage{amsopn}

\usepackage{enumerate}
\usepackage{graphics}
\usepackage{tikz}
\usepackage[all]{xy}

\usepackage{url}

\newtheorem{theorem}{Theorem}

\newtheorem {example}[theorem]{Example}

\newtheorem{property}{Property}

\newcommand{\sgraph}{G}

\newcommand{\weight}{w}
\newcommand{\neighbour}{N}
\newcommand{\products}{\mathcal{P}}

\newcommand{\snet}{\mathcal{S}}
\newcommand{\prodset}{P}

\newcommand{\srcnodes}{\mathit{source}}
\newcommand{\payoff}{p}
\newcommand{\strprofile}{s}

\newcommand{\inflset}{\mathcal{N}}

\newcommand{\constutil}{c_0}

\newcommand{\bigo}{O}

\newcommand{\bfe}[1]{\begin{bfseries}\emph{#1}\end{bfseries}\index{#1}}
\newcommand{\myra}{\mbox{$\:\rightarrow\:$}}
\newcommand{\fa}{\mbox{$\forall$}}

\newcommand{\LL}{\mbox{$\ldots$}}
\newcommand{\sse}{\mbox{$\:\subseteq\:$}}
\newcommand{\ES}{\emptyset}
\newcommand{\NI}{\noindent}

\newcommand{\II}{\vspace{2 mm}}
\newcommand{\HB}{\hfill{$\Box$}}

\newcommand{\oldbfe}[1]{\begin{bfseries}\emph{#1}\end{bfseries}}

\title{Social Network Games with Obligatory Product Selection}

\author{Krzysztof R. Apt
    \institute{%
      Centre for Mathematics and Computer Science (CWI), \\
      ILLC, University of Amsterdam, The Netherlands
    }
    \email{k.r.apt@cwi.nl}
\and
Sunil Simon
    \institute{%
      Centre for Mathematics and Computer Science (CWI)
  }
    \email{s.e.simon@cwi.nl}
}

\def\titlerunning{Social Network Games with Obligatory Product Selection}
\def\authorrunning{Krzysztof R. Apt \& Sunil Simon}
\begin{document}
\maketitle

\begin{abstract}
  Recently, we introduced in \cite{AM11} a model for product adoption
  in social networks with multiple products, where the agents,
  influenced by their neighbours, can adopt one out of several
  alternatives (products). To analyze these networks we introduce
  social network games in which product adoption is obligatory.
  
  We show that when the underlying graph is a simple cycle, there is a
  polynomial time algorithm allowing us to determine whether the game
  has a Nash equilibrium. In contrast, in the arbitrary case this
  problem is NP-complete. We also show that the problem of determining
  whether the game is weakly acyclic is co-NP hard.
  
  Using these games we analyze various types of paradoxes that can
  arise in the considered networks. One of them corresponds to the
  well-known Braess paradox in congestion games. In particular, we
  show that social networks exist with the property that by adding an
  additional product to a specific node, the choices of the nodes will
  unavoidably evolve in such a way that everybody is strictly worse
  off.
\end{abstract}

\section{Introduction}

Social networks became a huge interdisciplinary research area with
important links to sociology, economics, epidemiology, computer
science, and mathematics.  A flurry of numerous articles, notably the
influential \cite{Mor00}, and books, e.g., \cite{Jac08,EK10}, helped
to delineate better this area.  It deals with many diverse topics such
as epidemics, spread of certain patterns of social behaviour, effects
of advertising, and emergence of `bubbles' in financial markets.

Recently, we introduced in \cite{AM11} \emph{social networks with
  multiple products}, in which the agents (players), influenced by
their neighbours, can adopt one out of several alternatives
(products).  To study the situation when the product adoption is
obligatory we introduce here social network games in which product
adoption is obligatory.  An example of a studied situation is when a
group of people chooses an obligatory `product', for instance, an
operating system or a mobile phone provider, by taking into account
the choice of their friends.  The resulting games exhibit the
following \bfe{join the crowd} property:

\begin{quote}
  the payoff of each player weakly increases when more players choose his strategy.
\end{quote}
that we define more precisely in Subsection~\ref{subsec:sng}.

The considered games are a modification of the strategic games that we
recently introduced in \cite{SA12} and more fully in \cite{SA13}, in
which the product adoption was optional.  The insistence on product
selection leads to a different analysis and different results than the
ones reported there.  In particular, Nash equilibria need not exist
already in the case when the underlying graph is a simple cycle.  We
show that one can determine in polynomial time whether for such social
networks a Nash equilibrium exists.  We prove that for arbitrary
networks, determining whether a Nash equilibrium exists is
NP-complete.  We also show that for arbitrary networks and for
networks whose underlying graph has no source nodes, determining
whether the game is weakly acyclic is co-NP hard.

The considered social networks allow us to analyze various paradoxes
that were identified in the literature.  One example is the
\emph{paradox of choice} first formulated in \cite{Sch05}.  It has
been summarised in \cite[page 38]{Gig08} as follows:
\begin{quote}                                                          
  The more options one has, the more possibilities for experiencing
  conflict arise, and the more difficult it becomes to compare the
  options. There is a point where more options, products, and choices
  hurt both seller and consumer.
\end{quote}                                                  
The point is that consumers choices depend on their friends' and
acquaintances' preferences.

Another example is a `bubble' in a financial market, where a
decision of a trader to switch to some new financial product triggers
a sequence of transactions, as a result of which all traders involved
become worse off.

Such paradoxes are similar to the renowned Braess paradox which states
that in some road networks the travel time can actually increase when
new roads are added, see, e.g., \cite[pages 464-465]{NRTV07} and a
`dual' version of Braess paradox that concerns the removal of road
segments, studied in \cite{FKLS12,FKS12}. Both paradoxes were studied
by means of congestion games.  However, in contrast to congestion
games, Nash equilibria do not need to exist in the games we consider
here. Consequently, one needs to rely on different arguments.
Moreover, there are now two new types of paradoxes that correspond to
the situations when an addition, respectively, removal, of a product
can lead to a game with no Nash equilibrium.

For each of these four cases we present a social network that exhibits
the corresponding paradox. These paradoxes were identified first in
\cite{AMS13} in the case when the adoption of a product was not
obligatory. In contrast to the case here considered the existence of a
strongest paradox within the framework of \cite{AMS13} remains an open
problem.

\section{Preliminaries}
\label{sec:prelim}

\subsection{Strategic games}

A \bfe{strategic game} for $n > 1$ players, written as $(S_1, \ldots, S_n,
p_1, \ldots, p_n)$, consists of a non-empty set $S_i$ of
\bfe{strategies} and a \bfe{payoff function} $p_i : S_1 \times \cdots
\times S_n \myra \mathbb{R}$,
for each player $i$.

Fix a strategic game
$
G := (S_1, \ldots, S_n, p_1, \ldots, p_n).
$
We denote $S_1 \times \cdots \times S_n$ by $S$, 
call each element $s \in S$
a \bfe{joint strategy},
denote the $i$th element of $s$ by $s_i$, and abbreviate the sequence
$(s_{j})_{j \neq i}$ to $s_{-i}$. Occasionally we write $(s_i,
s_{-i})$ instead of $s$.  

We call a strategy $s_i$ of player $i$ a \bfe{best response} to a
joint strategy $s_{-i}$ of his opponents if $ \fa s'_i \in S_i
\  p_i(s_i, s_{-i}) \geq p_i(s'_i, s_{-i})$. We call a joint strategy
$s$ a \bfe{Nash equilibrium} if each $s_i$ is a best response to
$s_{-i}$.
Further, we call a strategy $s_i'$ of player $i$ a \bfe{better
  response} given a joint strategy $s$ if $p_i(s'_i, s_{-i}) >
p_i(s_i, s_{-i})$.

By a \bfe{profitable deviation} we mean a pair $(s,s')$ of joint
strategies such that $s' = (s'_i, s_{-i})$ for some $s'_i$ and
$p_i(s') > p_i(s)$. 
Following \cite{MS96}, an \bfe{improvement path} is a maximal sequence
of profitable deviations. Clearly, if an improvement path is finite,
then its last element is a Nash equilibrium.  A game is called
\bfe{weakly acyclic} (see \cite{You93,Mil96}) if for every joint
strategy there exists a finite improvement path that starts at it. In
other words, in weakly acyclic games a Nash equilibrium can be reached
from every initial joint strategy by a sequence of unilateral
deviations. Given two joint strategies $s$ and $s'$ we write
\begin{itemize}
\item $s > s'$ if for all $i$, $p_i(s) > p_i(s')$.
\end{itemize}
When $s > s'$ holds we say that $s'$ is \bfe{strictly worse} than $s$.

\subsection{Social networks}

We are interested in strategic games defined over a specific type of
social networks introduced in \cite{AM11} that we recall first.

Let $V=\{1,\ldots,n\}$ be a finite set of \bfe{agents} and $\sgraph=(V,E,\weight)$ 
a weighted directed graph with $\weight_{ij} \in [0,1]$ being the
weight of the edge $(i,j)$. 
Given a node $i$ of $G$, we denote by
$\neighbour(i)$ the set of nodes from which there is an incoming edge to $i$.
We call each $j \in \neighbour(i)$ a \oldbfe{neighbour} of $i$ in $G$.
We assume that for each node $i$ such that $\neighbour(i) \neq \ES$, $\sum_{j
\in \neighbour(i)} w_{ji} \leq 1$.
An agent $i \in V$ is said to be a
\bfe{source node} in $\sgraph$ if $\neighbour(i)=\emptyset$.
Given a (to be defined) network $\snet$ we denote by $\srcnodes(\snet)$ the set of
source nodes in the underlying graph $\sgraph$.

By a \bfe{social network} (from now on, just \bfe{network}) we mean a
tuple $\snet=(\sgraph,\products,\prodset,\theta)$, where 
\begin{itemize}
\item $G$ is a weighted directed graph, 

\item $\products$ is a finite set of alternatives or \bfe{products},

\item $\prodset$ is function that 
assigns to each agent $i$ a non-empty set of products $\prodset(i)$
from which it can make a choice, 

\item $\theta$ is a \bfe{threshold
  function} that for each $i \in V$ and $t \in \prodset(i)$ yields a
value $\theta(i,t) \in (0,1]$.
\end{itemize}


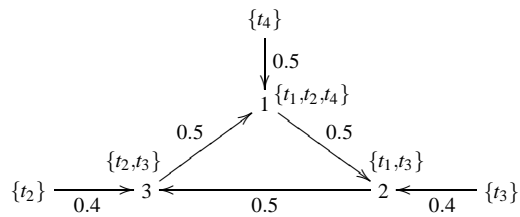
\begin{figure}[ht]
\centering
$
\def\objectstyle{\scriptstyle}
\def\labelstyle{\scriptstyle}
\xymatrix@R=20pt @C=30pt{
& &\{t_4\} \ar[d]^{0.5}\\
& &1 \ar[rd]^{0.5} \ar@{}[rd]^<{\{t_1,t_2,t_4\}}\\
\{t_2\} \ar[r]_{0.4} &3 \ar[ur]^{0.5} \ar@{}[ur]^<{\{t_2,t_3\}}& &2 \ar[ll]^{0.5} \ar@{}[lu]_<{\{t_1,t_3\}} &\{t_3\} \ar[l]^{0.4}\\
}$

\caption{\label{fig:socnet}A social network}
\end{figure}

\begin{example}
\label{ex:socnet}
\normalfont Figure \ref{fig:socnet} shows an example of a network. Let
the threshold be $0.3$ for all nodes. The set of products $\products$
is $\{t_1,t_2,t_3,t_4\}$, the product set of each agent is marked next to
the node denoting it and the weights are labels on the edges. Each
source node is represented by the unique product in its product set.
\HB
\end{example}

Given two social networks $\snet$ and $\snet'$ we say that $\snet'$ is
an \bfe{expansion} of $\snet$ if it results from adding a product to
the product set of a node in $\snet$.  We say then also that $\snet$
is a \bfe{contraction} of $\snet'$.

\subsection{Social network games}
\label{subsec:sng}

Next, introduce the strategic games over the social networks. They
form a modification of the games studied in \cite{SA12, SA13} in that we do
not admit a strategy representing the fact that a player abstains from
choosing a product.

Fix a network $\snet=(\sgraph,\products,\prodset,\theta)$. With each
network $\snet$ we associate a strategic game
$\mathcal{G}(\snet)$. The idea is that the agents simultaneously
choose a product.  Subsequently each node assesses his choice by
comparing it with the choices made by his neighbours.  Formally, we
define the game as follows:

\begin{itemize}
\item the players are the agents (i.e., the nodes),

\item the set of strategies for player $i$ is
$S_i :=\prodset(i)$,

\item For $i \in V$, $t \in
\prodset(i)$ and a joint strategy $\strprofile$, let
$
\inflset_i^t(\strprofile) :=\{j \in \neighbour(i) \mid s_j=t\},
$
i.e., $\inflset_i^t(\strprofile)$ is the set of neighbours of $i$ who adopted in $s$
the product $t$.

The payoff function is defined as follows, where $\constutil$ is some given in advance
positive constant:

\begin{itemize}
\item for $i \in \srcnodes(\snet)$,

$\payoff_i(\strprofile) := \constutil$,

\item for $i \not\in \srcnodes(\snet)$,

$\payoff_i(s) := \sum\limits_{j \in \inflset_i^t(\strprofile)} w_{ji}-\theta(i,t) \mbox{ , where } \strprofile_i=t \mbox{ and } t \in \prodset(i)$.

\end{itemize}

\end{itemize}

In the first case we assume that the payoff function for the source
nodes is constant only for simplicity.  The second case of the payoff
definition is motivated by the following considerations.  When agent
$i$ is not a source node, his `satisfaction' from a joint strategy
depends positively from the accumulated weight (read: `influence') of
his neighbours who made the same choice as him, and negatively from
his threshold level (read: `resistance') to adopt this product.  The
assumption that $\theta(i,t) > 0$ reflects the view that there is
always some resistance to adopt a product. 

We call these games
\bfe{social network games with obligatory product selection}, 
in short, \bfe{social network games}.

\begin{example}
\label{ex:payoff}
\normalfont Consider the network given in Example~\ref{ex:socnet} and
the joint strategy $\strprofile$ where each source node chooses the unique
product in its product set and nodes 1, 2 and 3 choose $t_2$, $t_3$
and $t_2$ respectively. The payoffs are then given as follows:

\begin{itemize}
\item for the source nodes, the payoff is the fixed constant $\constutil$,
\item $\payoff_1(\strprofile)=0.5-0.3=0.2$,
\item $\payoff_2(\strprofile)=0.4-0.3=0.1$,
\item $\payoff_3(\strprofile)=0.4-0.3=0.1$.
\end{itemize}

Let $\strprofile'$ be the joint strategy in which player 3 chooses
$t_3$ and the remaining players make the same choice as given in
$\strprofile$. Then $(\strprofile,\strprofile')$ is a profitable
deviation since $\payoff_3(\strprofile') > \payoff_3(\strprofile)$. In
what follows, we represent each profitable deviation by a node and a
strategy it switches to, e.g., $3:t_3$. Starting at $\strprofile$, the
sequence of profitable deviations $3:t_3, 1:t_4$ is an improvement
path which results in the joint strategy in which nodes 1, 2 and 3
choose $t_4$, $t_3$ and $t_3$ respectively and, as before, each source
node chooses the unique product in its product set.  \HB
\end{example}

By definition, the payoff of each player depends only on the
strategies chosen by his neighbours, so the social network games are
related to graphical games of \cite{KLS01}. However, the underlying
dependence structure of a social network game is a directed
graph. Further, note that these games satisfy the \bfe{join the crowd}
property that we define as follows:

\begin{quote}
Each payoff function $p_i$ depends only on the strategy chosen by player $i$ and
the set of players who also chose his strategy. Moreover, 
the dependence on this set is monotonic.
\end{quote}

The last qualification is exactly opposite to the definition of 
congestion games with player-specific payoff functions of~\cite{Mil96} 
in which the dependence on the above set is antimonotonic. 
That is, when more players choose the strategy of player $i$, then his payoff weakly decreases.

\section{Nash equilibria}

The first natural question we address is whether the social network
games have a Nash equilibrium.

\subsection{Simple cycles}

In contrast to the case of games studied in \cite{SA12} the answer is negative already
for the case when the underlying graph is a simple cycle.

\begin{example} \label{ex:nonash}
\label{ex:noNe1}
\rm
Consider the network given in Figure~\ref{fig:noNe-cycle}, where the product
set of each agent is marked next to the node denoting it and
the weights are all equal and put as labels on the edges.

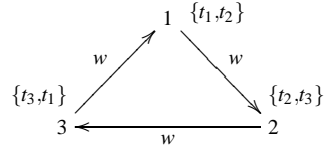
\begin{figure}[ht]
\centering
$
\def\objectstyle{\scriptstyle}
\def\labelstyle{\scriptstyle}
\xymatrix@R=30pt @C=30pt{
&1 \ar[rd]^{w} \ar@{}[rd]^<{\{t_1,t_2\}}\\
3 \ar[ur]^{w} \ar@{}[ur]^<{\{t_3,t_1\}}& &2 \ar[ll]^{w} \ar@{}[lu]_<{\{t_2,t_3\}} \\
}$
\caption{\label{fig:noNe-cycle}A simple cycle no Nash equilibrium}
\end{figure}

Let the thresholds be defined as follows:
$\theta(1,t_1)=\theta(2,t_2)=\theta(3,t_3)=r_1$ and
$\theta(1,t_2)=\theta(2,t_3)=\theta(3,t_1)=r_2$ where $r_1>r_2$. We
also assume that $w >r_1-r_2$. Hence for all $\strprofile_2$ and
$\strprofile_3$
\[\payoff_1(t_1,\strprofile_2,t_1) > \payoff_1(t_2,\strprofile_2,\strprofile_3) > \payoff_1(t_1,\strprofile_2,t_3)\]
and similarly for the payoff functions $\payoff_2$ and $\payoff_3$. So
it is more profitable for player $i$ to adopt strategy $t_i$ provided
its neighbour also adopts $t_i$.

It is easy to check that the game associated with this network has no
Nash equilibrium. Indeed, here is the list of all the joint strategies,
where we underline the strategy that is not a best response to the
choice of other players: $(t_1,\underline{t_2},t_1)$,
$(t_1,\underline{t_2},t_3)$, $(t_1,t_3,\underline{t_1})$,
$(\underline{t_1}, t_3, t_3)$, $(\underline{t_2},t_2,t_1)$,
$(t_2,t_2,\underline{t_3})$, $(\underline{t_2},t_3,t_1)$,
$(t_2,\underline{t_3},t_3)$.
\HB
\end{example}

This example can be easily generalized to the case of an arbitrary
simple cycle. Below, $i \oplus 1$ and $i \ominus 1$ stand for addition
and subtraction defined cyclically over the set $\{1,\ldots,n\}$. So
$n \oplus 1=1$ and $1 \ominus 1=n$.  Indeed, consider a social network
with $n$ nodes that form a simple cycle and assume that each player
$i$ has strategies $t_i$ and $t_{i \oplus 1}$. Choose for each player
$i$ the weights $w_{i \ominus 1 \: i}$ and the threshold function
$\theta(i,t)$ so that
\[
w_{i \ominus 1 \: i} - \theta(i,t_i) > - \theta(i,t_{i \oplus 1}) > - \theta(i,t_i),
\]
so that (we put on first two positions, respectively, the strategies
of players $i \ominus 1$ and $i$, while the last argument is a joint
strategy of the remaining $n-2$ players)
\[
p_i(t_i, t_i, s) > p_i(t', t_{i \oplus 1}, s') > p_i(t_{i \ominus 1}, t_i, s''),
\]
where $t', s, s'$ and $s''$ are arbitrary.  It is easy to check then
that the resulting social network game has no Nash equilibrium.

A natural question is what is the complexity of determining whether a
Nash equilibrium exists.  First we consider this question for the
special case when the underlying graph is a simple cycle.

\begin{theorem} \label{thm:cycle}
Consider a network $\snet$ whose underlying graph is a simple cycle.
It takes $\bigo(n \cdot |\products|^4)$ time to decide whether the game
$\mathcal{G}(\snet)$ has a Nash equilibrium.
\end{theorem}
\begin{proof}
Suppose $\snet=(\sgraph,\products,\prodset,\theta)$.  When the
underlying graph of $\snet$ is a simple cycle, the concept of a best
response of player $i \oplus 1$ to a strategy of player $i$ is
well-defined.  Let
\[
R_i := \{(t_i, t_{i \oplus 1}) \mid t_i \in \prodset(i), t_{i \oplus
  1} \in \prodset(i \oplus 1), \mbox{$t_{i \oplus 1}$ is a best
  response to $t_i$}\},
\]
\[
I := \{(t,t) \mid t \in \products \},
\]
and let $\circ$ stand for the composition of binary relations.

The question whether $\mathcal{G}(\snet)$ has a Nash equilibrium is
then equivalent to the problem whether there exists a sequence $a_1,
..., a_n$ such that $(a_1, a_2) \in R_1, ..., (a_{n-1}, a_n) \in R_{n-1}, (a_n, a_1) \in R_n$.  In
other words, is $(R_1 \circ \dots \circ R_n) \cap I$ non-empty?

To answer this question we first construct successively $n-1$ compositions
$R_1 \circ R_2$, $(R_1 \circ R_2) \circ R_3$, $\dots$, $(\dots (R_1 \circ R_2)
\dots \circ R_{n-1}) \circ R_n$.

Each composition construction can be carried out in $|\products|^4$
steps.  Indeed, given two relations $A, B \sse \products \times
\products$, to compute their composition $A \circ B$ requires for each
pair $(a,b) \in A$ to find all pairs $(c,d) \in B$ such that $b = c$.
Finally, to check whether the intersection of $R_1 \circ \dots \circ
R_n$ with $I$ is non-empty requires at most $|\products|$ steps.

So to answer the original question takes
$\bigo(n \cdot |\products|^4)$ time.
\end{proof}

Note that this proof applies to any strategic game in which there is a
reordering of players $\pi(1), \dots, \pi(n)$ such that the payoff of
player $\pi(i)$ depends only on his strategy and the strategy chosen
by player $\pi(i \ominus i)$.


It is worthwhile to note that for the case of simple cycles, the
existence of Nash equilibrium in the associated social network game
does not imply that the game is weakly acyclic.

\begin{figure}[ht]
\centering
\begin{tabular}{ccc}
$
\def\objectstyle{\scriptstyle}
\def\labelstyle{\scriptstyle}
\xymatrix@R=30pt @C=30pt{
&1 \ar[rd]^{w} \ar@{}[rd]^<{\{t_1,t_2,t_4\}}\\
3 \ar[ur]^{w} \ar@{}[ur]^<{\{t_3,t_1,t_4\}}& &2 \ar[ll]^{w} \ar@{}[lu]_<{\{t_2,t_3,t_4\}} \\
}$
& &
$                                                                               
\def\objectstyle{\scriptstyle}                                                  
\def\labelstyle{\scriptstyle}                                                   
\xymatrix@W=8pt @C=15pt @R=15pt{                                                
(t_1,t_3,\underline{t_1})\ar@{=>}[r]& (\underline{t_1},t_3,t_3)\ar@{=>}[r]& (t_2,\underline{t_3}, t_3)\ar@{=>}[d]\\
(t_1, \underline{t_2}, t_1)\ar@{=>}[u]& (\underline{t_2},t_2,t_1)\ar@{=>}[l]& (t_2,t_2,\underline{t_3})\ar@{=>}[l]\\
}$\\
(a) & & (b)
\end{tabular}
\caption{\label{fig:notWA}A simple cycle and an infinite improvement path}
\end{figure}
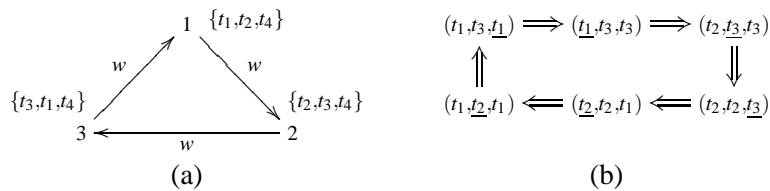

\begin{example}
\label{ex:Ne-noWA}
\normalfont Consider the network in Figure~\ref{fig:notWA}(a) which is
a modification of the network in Figure~\ref{fig:noNe-cycle}. We add a
new product $t_4$ to the product set of all the nodes $i$ with
$\theta(i,t_4) > r_1$. We also assume that $w -\theta(i,t_4) >
-r_2$. Then the joint strategy $(t_4,t_4,t_4)$ is a Nash
equilibrium. However, Figure~\ref{fig:notWA}(b) shows the unique
improvement path starting in $(t_1,t_3,t_1)$ which is infinite. For
each joint strategy in the figure, we underline the strategy that is
not a best response. This shows that the game is not weakly acyclic.
\HB
\end{example}

In Section \ref{sec:WA}  we shall study the complexity of checking
whether a social network game is weakly acyclic.


\subsection{Arbitrary social networks}
In this section we establish two results which show that deciding
whether a social network has a Nash equilibrium is computationally
hard.
\begin{theorem}
\label{thm:Ne-npcomplete}
Deciding whether for a social network $\snet$ the game
$\mathcal{G}(\snet)$ has a Nash equilibrium is NP-complete.
\end{theorem}

To prove the result we first construct another example of a social
network game with no Nash equilibrium and then use it to determine the
complexity of the existence of Nash equilibria.

\begin{example} \label{ex:noNe2}
\rm Consider the network given in Figure~\ref{fig:noNe2}, where the
product set of each agent is marked next to the node denoting it and
the weights are labels on the edges. Nodes with a unique product in
the product set is simply represented by the product.
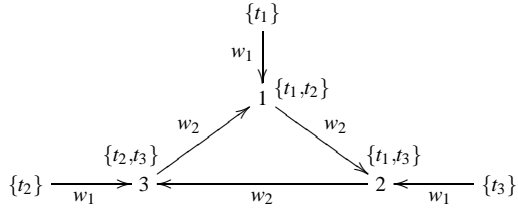
\begin{figure}[ht]
\centering
$
\def\objectstyle{\scriptstyle}
\def\labelstyle{\scriptstyle}
\xymatrix@R=20pt @C=30pt{
& & & \{t_1\} \ar[d]_{w_1}\\
& & &1 \ar[rd]^{w_2} \ar@{}[rd]^<{\{t_1,t_2\}}\\
&\{t_2\} \ar[r]_{w_1} &3 \ar[ur]^{w_2} \ar@{}[ur]^<{\{t_2,t_3\}}& &2 \ar[ll]^{w_2} \ar@{}[lu]_<{\{t_1,t_3\}} &\{t_3\} \ar[l]^{w_1}\\
}$

\caption{\label{fig:noNe2}A network with no Nash equilibrium}
\end{figure}

We assume that each threshold is a constant
$\theta$, where $ \theta < w_1 < w_2.  $ So it is more profitable to a
player residing on a triangle to adopt the product adopted by his
neighbour residing on a triangle than by the other neighbour.

The game associated with this network has no Nash equilibrium. It
suffices to analyze the joint strategies involving nodes 1, 2 and 3
since the other nodes have exactly one product in their product
sets. Here we provide a listing of all such joint strategies, where we
underline the strategy that is not a best response to the choice of
other players: $(\underline{t_1}, t_1, t_2)$, $(t_1, t_1,
\underline{t_3})$, $(t_1, t_3, \underline{t_2})$, $(t_1,
\underline{t_3}, t_3)$, $(t_2, \underline{t_1}, t_2)$, $(t_2,
\underline{t_1}, t_3)$, $(t_2, t_3, \underline{t_2})$,
$(\underline{t_2}, t_3, t_3)$.  In contrast, what will be of relevance
in a moment, if we replace $\{t_1\}$ by $\{t_1'\}$, then the
corresponding game has a Nash equilibrium, namely the joint strategy
corresponding to the triple $(t_2,t_3,t_3)$.  
\HB
\end{example}

\noindent{\it Proof of Theorem \ref{thm:Ne-npcomplete}:} As in \cite{AM11}, to
show NP-hardness, we use a reduction from the NP-complete PARTITION
problem, which is: given $n$ positive rational numbers
$(a_1,\ldots,a_n)$, is there a set $S$ such that $\sum_{i\in S} a_i =
\sum_{i\not\in S} a_i$?  Consider an instance $I$ of PARTITION.
Without loss of generality, suppose we have normalised the numbers so
that $\sum_{i=1}^n a_i = 1$. Then the problem instance sounds: is
there a set $S$ such that $\sum_{i\in S} a_i = \sum_{i\not\in S} a_i =
\frac{1}{2}$?
  
To construct the appropriate network we employ the networks given in
Figure~\ref{fig:noNe2} and in Figure~\ref{fig:partition}, where for
each node $i\in\{1,\ldots,n\}$ we set $w_{i a} = w_{i b} = a_i$, and
assume that the thresholds of the nodes $a$ and $b$ are constant and
equal $\frac12$. 

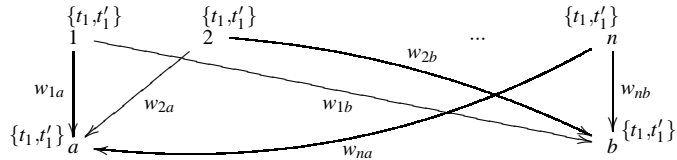
\begin{figure}[ht]
\centering
$
\def\objectstyle{\scriptstyle}
\def\labelstyle{\scriptstyle}
\xymatrix@W=10pt @R=30pt @C=35pt{
1 \ar@{}[r]^<{\{t_1,t_1'\}} \ar[d]_{w_{1a}} \ar[rrrrd]_{w_{1b}}& 2 \ar@{}[r]^<{\{t_1,t_1'\}} \ar[ld]^{w_{2a}} \ar@/^0.7pc/[rrrd]^{w_{2b}}&  &\cdots &n \ar@{}[l]_<{\{t_1,t_1'\}} \ar@/^1.5pc/[lllld]^{w_{na}} \ar[d]^{w_{nb}}\\
a \ar@{}[u]^<{\{t_1,t_1'\}}  & & & &b \ar@{}[u]_<{\{t_1,t_1'\}} \\
}$
\caption{\label{fig:partition}A network related to the PARTITION problem}
\end{figure}

To finalize the construction we use two copies of the network given in
Figure~\ref{fig:noNe2}, one unchanged and the other in which the
product $t_1$ is replaced everywhere by $t'_1$, and
construct the desired network
$\snet$ by identifying with the node marked by $\{t_1\}$ in the
network from Figure~\ref{fig:noNe2}, the node $a$ of the network from
Figure~\ref{fig:partition}  and with the node marked
by $\{t'_1\}$ in the modified version of the network from
Figure~\ref{fig:noNe2} the node $b$.

Suppose that a solution to the considered instance of the PARTITION
problem exists, i.e., for some set $S \subseteq \{1, \ldots, n\}$ we
have $\sum_{i\in S} a_i = \sum_{i\not\in S} a_i = \frac{1}{2}$.
Consider the game $\mathcal{G}(\snet)$ and the joint strategy
$\strprofile$ formed by the following strategies:

\begin{itemize}

\item $t_1$ assigned to each node $i \in S$ in the network from
  Figure~\ref{fig:partition},

\item $t'_1$ assigned to each node $i \in \{1, \ldots, n\} \setminus
  S$ in the network from Figure~\ref{fig:partition},

\item $t_1'$ assigned to the nodes $a$ and $t_1$ to the node $b$,

\item $t_2$ assigned to node 1 and $t_3$ assigned to the nodes 2, 3
  in both versions of the networks from Figure~\ref{fig:noNe2},

\item $t_2$ and $t_3$ assigned respectively to the nodes marked by
  $\{t_2\}$ and $\{t_3\}$.

\end{itemize}

We claim that $\strprofile$ is a Nash equilibrium. Consider first the
player (i.e., node) $a$. The accumulated weight of its neighbours who
chose strategy $t_1'$ is $\frac12$. Therefore, the payoff for $a$ in
the joint strategy $\strprofile$ is $0$. The accumulated weight of its
neighbours who chose strategy $t_1$ is $\frac12$, as well. Therefore
$t_1'$ is indeed a best response for player $a$ as both strategies
yield the same payoff.  For the same reason, $t_1$ is a best response
for player $b$.  The analysis for the other nodes is straightforward.

Conversely, suppose that a strategy profile $\strprofile$ is a Nash
equilibrium in $\mathcal{G}(\snet)$. From Example \ref{ex:noNe2} it
follows that $\strprofile_a=t_1'$ and $\strprofile_b=t_1$. This
implies that $t_1'$ is a best response of node $a$ to
$\strprofile_{-a}$ and therefore $ \sum_{i \in \{1, \ldots, n\} \mid
  s_i = t_1'} w_{i a} \geq \sum_{i \in \{1, \ldots, n\} \mid s_i =
  t_1} w_{i a}$. By a similar reasoning, for node $b$ we have $
\sum_{i \in \{1, \ldots, n\} \mid s_i = t_1} w_{i b} \geq \sum_{i \in
  \{1, \ldots, n\} \mid s_i = t_1'} w_{i b}$. Since $\sum_{i=1}^n a_i
= 1$ and for $i\in\{1,\ldots,n\}$, $w_{i a} = w_{i b} = a_i$, and $s_i
\in \{t_1, t'_1\}$ we have for $S := \{i \in \{1, \ldots, n\} \mid s_i
= t_1\}$, $\sum_{i\in S} a_i = \sum_{i\not\in S} a_i$.  In other
words, there exists a solution to the considered instance of the
partition problem.  \HB

\begin{theorem}
\label{thm:Ne-nosource-npcomplete}
For a network $\snet$ whose underlying graph has no source nodes,
deciding whether the game $\mathcal{G}(\snet)$ has a Nash equilibrium
is NP-complete.
\end{theorem}
\begin{proof}
The proof extends the proof of the above theorem. Given an instance of
the PARTITION problem we use the following modification of the
network. We `twin' each node $i \in \{1,\ldots,n\}$ in
Figure~\ref{fig:partition} with a new node $i'$ with the product set
$\{t_1,t_1'\}$, by adding the edges $(i,i')$ and $(i',i)$. We also
`twin' nodes marked $\{t_2\}$ and $\{t_3\}$ in Figure~\ref{fig:noNe2}
with new nodes with the product set $\{t_2\}$ and $\{t_3\}$
respectively. Additionally, we choose the weights on the new edges
$w_{ii'}$, $w_{i'i}$ and the corresponding thresholds so that when $i$
and $i'$ adopt a common product, their payoff is positive. Then the
underlying graph of the resulting network does not have any source
nodes and the above proof remains valid for this new network.
\end{proof}

\section{Weakly acyclic games}
\label{sec:WA}
In this section we study the complexity of checking whether a social
network game is weakly acyclic.  We establish two results that are
analogous to the ones established in \cite{SA13} for the case of
social networks in which the nodes may decide not to choose any
product. The proofs are based on similar arguments though the details
are different.
 
\begin{theorem}\label{thm:acyclic}
For an arbitrary network $\snet$, deciding whether the game
$\mathcal{G}(\snet)$ is weakly acyclic is co-NP hard.
\end{theorem}
\begin{proof}
We again use an instance of the PARTITION problem in the form of $n$
positive rational numbers $(a_1,\LL,a_n)$ such that $\sum_{i=1}^n a_i
= 1$.
Consider the network given in Figure~\ref{fig:ufip-hard}. For each
node $i\in\{1,\LL,n\}$ we set $P(i) = \{t_1, t_2\}$. The product set
for the other nodes are marked in the figure. As before, we set
$w_{i a} = w_{i b} = a_i$.

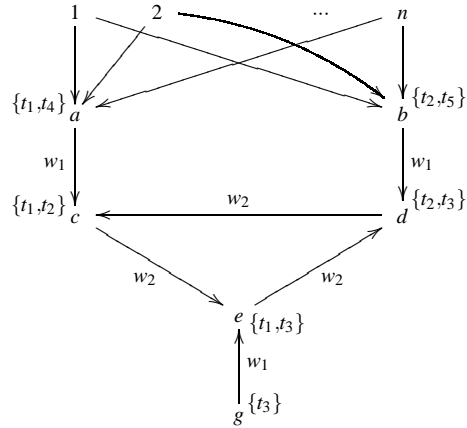
\begin{figure}[ht]
\centering
$
\def\objectstyle{\scriptstyle}
\def\labelstyle{\scriptstyle}
\xymatrix@W=10pt @R=27pt @C=15pt{
1 \ar[d]_{} \ar[rrrrd]_{}& 2 \ar[ld]^{} \ar@/^0.7pc/[rrrd]^{}& & \cdots &n \ar[lllld]^{} \ar[d]^{} \\
a \ar@{}[u]^<{\{t_1,t_4\}} \ar[d]_{w_1} & & & &b \ar@{}[u]_<{\{t_2,t_5\}} \ar[d]^{w_1} \\
c  \ar[rrd]_{w_2} \ar@{}[u]^<{\{t_1,t_2\}}& & & &d  \ar[llll]_{w_2} \ar@{}[u]_<{\{t_2,t_3\}}\\
& &e \ar[rru]_{w_2} \ar@{}[rru]_<{\{t_1,t_3\}}\\
& &g \ar[u]_{w_1} \ar@{}[u]_<{\{t_3\}}\\
}$
\caption{\label{fig:ufip-hard}A network related to weakly acyclic games}
\end{figure}

Since for all $i \in \{1,\ldots,n\}$, $a_i$ is rational, it has the
form $a_i = \frac{l_i}{r_i}$. Let $\tau=\frac{1}{2\cdot r_1 \cdot
  \ldots \cdot r_n}$. 
The following property holds.

\begin{property}
\label{partition-tau}
Given an instance $(a_1,\ldots,a_n)$ of the PARTITION problem and
$\tau$ defined as above, for all $S \subseteq \{1,\ldots,n\}$
\begin{enumerate}[(i)]
\item if $\sum_{i \in S} a_i < \frac12$, then $\sum_{i \in S} a_i \leq \frac12 -\tau$,
\item if $\sum_{i \in S} a_i > \frac12$, then $\sum_{i \in S} a_i \geq \frac12 +\tau$.
\end{enumerate}
\end{property}
\begin{proof}
By definition, each $a_i$ and $\frac12$ is a multiple of $\tau$. Thus
$\sum_{i \in S} a_i = x \cdot \tau$ and $\frac12 = y \cdot \tau$ where
$x$ and $y$ are integers. \\
\noindent {\it (i)} If $x \cdot \tau < y \cdot \tau$, then $x \cdot
\tau \leq (y-1) \cdot \tau$. Therefore $\sum_{i \in S} a_i \leq \frac12 -
\tau$.\\
\noindent The proof of {\it (ii)} is analogous.
\end{proof}
Note that given $(a_1,\ldots,a_n)$, $\tau$ can be defined in
polynomial time. Let the thresholds be defined as follows:
$\theta(a,t_1)=\theta(b,t_2)=\frac12$ and $0 <
\theta(a,t_4)=\theta(b,t_5) < \tau$. The threshold for nodes $c,d$ and
$e$ is a constant $\theta_1$ such that $\theta_1 < w_1 < w_2$. Thus,
like in the network in Figure~\ref{fig:noNe2}, it is more profitable
to a player residing on a triangle to adopt the product adopted by his
neighbour residing on a triangle than by the other neighbour.

Suppose that a solution to the considered instance of the PARTITION
problem exists. That is, for some set $S \sse \{1, \LL, n\}$ we have
$\sum_{i\in S} a_i = \sum_{i\not\in S} a_i = \frac{1}{2}$.  In the
game $\mathcal{G}(\snet)$, take the joint strategy $s$ formed by the
following strategies:

\begin{itemize}

\item $t_1$ assigned to each node $i \in S$ and the nodes $a$ and $c$,

\item $t_2$ assigned to each node $i \in \{1, \LL, n\} \setminus S$
and the nodes $b$ and $d$,

\item $t_3$ assigned to the nodes $e$ and $g$.

\end{itemize}

Any improvement path that starts in this joint strategy will not
change the strategies assigned to the nodes $a, b$ and $g$. So if such
an improvement path terminates, it produces a Nash equilibrium in the
game associated with the network given in Figure~\ref{fig:noNe2} of
Example~\ref{ex:noNe2}. But we argued that this game does not have a
Nash equilibrium. Consequently, there is no finite improvement path in
the game $\mathcal{G}(\snet)$ that starts in the above joint strategy
and therefore $\mathcal{G}(\snet)$ is not weakly acyclic.

Now suppose that the considered instance of the PARTITION problem does
not have a solution. Then we show that the game $\mathcal{G}(\snet)$
is weakly acyclic.  To this end, we order the nodes of $\snet$ as
follows (note the positions of the nodes $c, d$ and $e$): $1, 2, \LL,
n, g, a, b, c, e, d$. Given a joint strategy, consider an improvement
path in which at each step the first node in the above list that did
not select a best response switches to a best response. After at most
$n$ steps the nodes $1, 2, \LL, n$ all selected a product $t_1$ or
$t_2$. Let $\strprofile$ be the resulting joint strategy. 

First suppose that $\sum_{i \in \{1, \LL, n\} \mid s_i = t_1} w_{i a}
> \frac12$. This implies that $\sum_{i \in \{1, \LL, n\} \mid s_i =
  t_2} w_{i b} < \frac12$. By Property~\ref{partition-tau}, $\sum_{i
  \in \{1, \LL, n\} \mid s_i = t_2} w_{i b} \leq \frac12 - \tau$.
The payoff of the node $b$ depends only on the choices made by the source
nodes $1, 2, \ldots, n$, so we have $\payoff_b(t_2,\strprofile_{-b}) \leq
-\tau$. Since $\theta(b,t_5) < \tau$, we also have
$\payoff_b(t_5,\strprofile_{-b}) > -\tau$ and therefore $t_5$ is a
best response for node $b$. Let $\strprofile^b$ be the resulting
strategy in which node $b$ selects $t_5$. Consider the prefix of $\xi$
starting at $\strprofile^b$ (call it $\xi^b$). We argue that in $\xi^b$,
$t_2$ is never a better response for node $d$. Suppose that
$\strprofile^b_d=t_3$. We have the following two cases:
\begin{itemize}
\item $\strprofile^b_e=t_3$: then $\payoff_d(\strprofile^b) = w_2 -
  \theta_1$ and so $t_3$ is the best response for node $d$.
\item $\strprofile^b_e=t_1$: then $\payoff_d(\strprofile^b) =
  -\theta_1$ and if node $d$ switches to $t_2$ then $\payoff_d(t_2,
  \strprofile^b_{-b}) =-\theta_1$ (since
  $\strprofile^b_b=t_5$). Thus $t_2$ is not a better response.
\end{itemize}
Using the above observation, we conclude that there exists a suffix
of $\xi^b$ (call it $\xi^d$) such that node $d$ never chooses
$t_2$. This means that in $\xi^d$ the unique best response for node
$c$ is $t_1$ and for node $e$ is $t_1$. This shows that $\xi^d$ is
finite and hence $\xi$ is finite, as well.

The case when $\sum_{i \in \{1, \LL, n\} \mid s_i = t_2} w_{i b} >
\frac12$ is analogous with all improvement paths terminating in a
joint strategy where node $a$ chooses $t_4$ and node $c$
chooses $t_2$.
\end{proof}

\begin{theorem}\label{thm:acyclic2}
  For a network $\snet$ whose underlying graph has no source nodes,
  deciding whether the game $\mathcal{G}(\snet)$ is weakly acyclic is
  co-NP hard.
\end{theorem}

\begin{proof}
The proof extends the proof of the above theorem. Given
an instance of the PARTITION problem we use the following modification
of the network given in Figure~\ref{fig:ufip-hard}.  We `twin' each
node $i \in \{1, \LL, n\}$ with a new node $i'$, also with the product
set $\{t_1, t_2\}$, by adding the edges $(i,i')$ and $(i',i)$. We also
`twin' the node $g$ with a new node $g'$, also with the product set
$\{t_3\}$, by adding the edges $(g,g')$ and $(g',g)$.   Additionally, we choose
the weights $w_{i i'}$ and $w_{i' i}$ and the corresponding thresholds so that when
$i$ and $i'$ adopt a common product, their payoff is positive.

Suppose that a solution to the considered instance of the PARTITION
problem exists. Then we extend the joint strategy considered in the
proof of Theorem~\ref{thm:acyclic} by additionally assigning $t_1$ to
each node $i'$ such that $i \in S$, $t_2$ to each node $i'$ such that
$i \in \{1, \LL, n\} \setminus S$ and $t_3$ to the node $g'$.  Then,
as before, there is no finite improvement path starting in this joint
strategy, so $\mathcal{G}(\snet)$ is not weakly acyclic.

Suppose now that no solution to the considered instance of the
PARTITION problem exists. Take the
following order of the nodes of $\snet$:
\[
1, 1', 2, 2', \LL, n, n', g, g', a, b, c, e, d,
\]
and as in the previous proof, given a joint strategy, we consider an
improvement path $\xi$ in which at each step the first node in the
above list that did not select a best response switches to a best
response.

Note that each node from the list $1, 1', 2, 2', \LL, n, n', g, g'$ is
scheduled at most once.  So there exists a suffix of $\xi$ in which
only the nodes $a, b, c, e, d$ are scheduled.  Using now the argument
given in the proof of Theorem~\ref{thm:acyclic} we conclude that there
exists a suffix of $\xi$ that is finite. This proves that
$\mathcal{G}(\snet)$ is weakly acyclic.
\end{proof}

\section{Paradoxes}

In \cite{AMS13} we identified various paradoxes in social networks
with multiple products and studied them using the social network games
introduced in \cite{SA12}. Here we carry out an analogous analysis for
the case when the product selection is obligatory. This qualification,
just like in the case of social network games, substantially changes
the analysis. We focus on the main four paradoxes that we successively
introduce and analyze.

\subsection{Vulnerable networks}

The first one is the following.  We say that a social network $\snet$
is \bfe{vulnerable} if for some Nash equilibrium $s$ in
$\mathcal{G}(\snet)$, an expansion $\snet'$ of $\snet$ exists such
that each improvement path in $\mathcal{G}(\snet')$ leads from $s$ to
a joint strategy $s'$ which is a Nash equilibrium both in
$\mathcal{G}(\snet')$ and $\mathcal{G}(\snet)$ such that $s
>s'$. So the newly added product triggers a sequence of changes that
unavoidably move the players from one Nash equilibrium to another one
that is strictly worse for everybody.

The following example shows that vulnerable networks exist. Here
and elsewhere the relevant expansion is depicted by means of a product
and the dotted arrow pointing to the relevant node.

\begin{example} 
\label{ex:fas}
\rm

Consider the directed graph given in Figure~\ref{fig:fas}, in which the product set of each node
is marked next to it. 
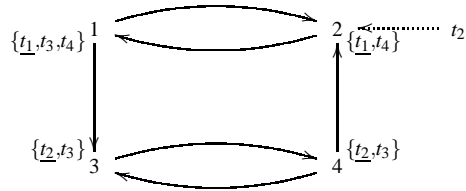
\begin{figure}[ht]
\centering
$
\def\objectstyle{\scriptstyle}
\def\labelstyle{\scriptstyle}
\xymatrix@W=10pt @R=40pt @C=30pt{
1 \ar@{}[d]_<{\{\underline{t_1},t_3,t_4\}} \ar[d]_{} \ar@/^0.7pc/[rr]^{}& &2 \ar@/^0.7pc/[ll]^{} \ar@{}[d]^<{\{\underline{t_1},t_4\}}& t_2 \ar@{..>}[l] \\
3 \ar@{}[u]^<{\{\underline{t_2},t_3\}} \ar@/^0.7pc/[rr]^{}& &4 \ar@/^0.7pc/[ll]^{} \ar@{}[u]_<{\{\underline{t_2},t_3\}} \ar[u]_{} \\
}$ 
\caption{\label{fig:fas}A directed graph}
\end{figure}

We complete it to the desired social network below.
Let `$\_$' stand for an arbitrary strategy of the relevant player.
We stipulate that
\II

$p_2(\_, t_2, \_, t_2) > p_2(t_1, t_1, \_, \_)$,

$p_1(t_3, t_2, \_, \_) > p_1(t_1, t_2, \_, \_) > p_1(t_4, t_2, \_, \_)$,

$p_3(t_3, \_, t_3, \_) > p_3(\_, \_,t_2, t_2)$,

$p_4(\_, \_, t_3, t_3) > p_4(\_, \_,t_3, t_2)$,

$p_2(\_, t_4, \_, \_) > p_2(\_, t_2, \_, t_3)$,

$p_1(t_4, t_4, \_, \_) > p_1(t_3, \_, \_, \_) > p_1(t_1, t_4, \_, \_)$,
\II

\NI
so that
$2: t_2, 1: t_3, 3: t_3, 4: t_3, 2: t_4, 1: t_4$
is a unique improvement path that starts in $(t_1, t_1, t_2, t_2)$
and ends in $(t_4, t_4, t_3, t_3)$.

Additionally we stipulate that 
\II

$p_1(t_1, t_1, \_, \_) > p_1(t_4, t_4, \_, \_)$,

$p_2(t_1, t_1, \_, \_) > p_2(t_4, t_4, \_, \_)$,

$p_3(\_, \_, t_2, t_2) > p_3(\_, \_, t_3, t_3)$,

$p_4(\_, \_, t_2, t_2) > p_4(\_, \_, t_3, t_3)$,
\II

\NI
so that
$(t_1, t_1, t_2, t_2) >_s (t_4, t_4, t_3, t_3)$.

These requirements entail constraints on the weights and thresholds that are for instance realized
by

$
w_{1 2} = 0, \ w_{2 1} = 0.2, \ w_{4 2} = 0.3, \ w_{1 3} = 0.2, \ w_{3 4} = 0.2, \ w_{4 3} = 0,
$

\NI
and

$
\theta(1,t_1) = 0.2, \ \theta(1,t_3) = 0.1, \ \theta(1,t_4) = 0.3, \ \theta(2,t_1) = 0.1, \ \theta(2,t_2) = 0.3,
$ 

$
\theta(2,t_4) = 0.2, \ \theta(3,t_2) = 0.1, \ \theta(3,t_3) = 0.2, \ \theta(4,t_2) = 0.1, \ \theta(4,t_3) = 0.2.
$
\HB
\end{example}

It is useful to note that in the setup of \cite{AMS13}, in which for each player the `abstain' strategy is allowed,
it remains an open problem whether vulnerable networks (called there because of various other alternatives $\fa s$-vulnerable
networks) exist.

\subsection{Fragile networks}

Next, we consider the following notion.  We say that a social network
$\snet$ is \bfe{fragile} if $\mathcal{G}(\snet)$ has a Nash
equilibrium while for some expansion $\snet'$ of $\snet$,
$\mathcal{G}(\snet')$ does not. The following example shows that
fragile networks exist.

\begin{example}
\label{ex:fragile}
\normalfont
Consider the network $\snet$ given in Figure~\ref{fig:fragile}, where
the product set of each node is marked next to it. 

\begin{figure}[ht]
\centering
$
\def\objectstyle{\scriptstyle}
\def\labelstyle{\scriptstyle}
\xymatrix@R=30pt @C=30pt{
&1 \ar[rd]^{w} \ar@{}[ld]_<{\{t_2\}} & t_1 \ar@{-->}[l]^{}\\
3 \ar[ur]^{w} \ar@{}[ur]^<{\{t_3,t_1\}}& &2 \ar[ll]^{w} \ar@{}[lu]_<{\{t_2,t_3\}} \\
}$
\caption{\label{fig:fragile}A fragile network}
\end{figure}
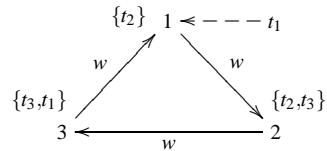

Let the thresholds
be defined as follows: $\theta(2,t_2)=\theta(3,t_3)=r_1$
and $\theta(1,t_2)=\theta(2,t_3)=\theta(3,t_1)=r_2$ where
$r_1>r_2$. We also assume that $w >r_1-r_2$.

Consider the joint strategy $\strprofile$, in which nodes 1, 2 and 3
choose $t_2$, $t_2$ and $t_1$ respectively. It can be verified that
$\strprofile$ is a Nash equilibrium in $\mathcal{G}(\snet)$. Now
consider the expansion $\snet'$ of $\snet$ in which product $t_1$ is
added to the product set of node 1 and let $\theta(1,t_1)=r_1$. Then
$\snet'$ is the network in Example~\ref{ex:noNe1} which, as we saw,
does not have a Nash equilibrium.  \HB
\end{example}

\subsection{Inefficient networks}

We say that a social network $\snet$ is \bfe{inefficient} if for some
Nash equilibrium $s$ in $\mathcal{G}(\snet)$, a contraction $\snet'$
of $\snet$ exists such that each improvement path in
$\mathcal{G}(\snet')$ starting in $\strprofile$ leads to a joint
strategy $\strprofile'$ which is a Nash equilibrium both in
$\mathcal{G}(\snet')$ and $\mathcal{G}(\snet)$ such that $s' > s$.  We
note here that if the contraction was created by removing a product
from the product set of node $i$, we impose that any improvement path
in $\mathcal{G}(\snet')$, given a starting joint strategy from
$\mathcal{G}(\snet)$, begins by having node $i$ making a choice (we
allow any choice from his remaining set of products as an improvement
move). Otherwise the initial payoff of node $i$ in
$\mathcal{G}(\snet')$ is not well-defined.

\begin{example}
\label{ex:inefficient}
\rm

We exhibit in Figure \ref{fig:inefficient} an example of an
inefficient network. The weight of each edge is assumed to be $w$, and
we also have the same product-independent threshold, $\theta$, for all
nodes, with $w> \theta$.

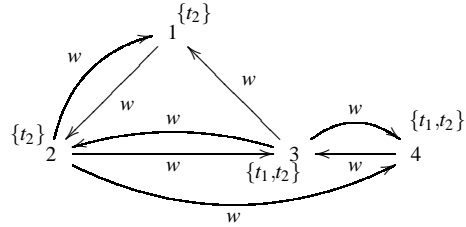
\begin{figure}[ht]
\centering
$
\def\objectstyle{\scriptstyle}
\def\labelstyle{\scriptstyle}
\xymatrix@W=10pt @R=35pt @C=30pt{
& 1 \ar[dl]^{w} \ar@{}[rr]^<{\{t_2\}}& & \\
2 \ar[rr]_{w} \ar@/^1pc/[ru]^{w} \ar@/_1.6pc/[rrr]_{w} \ar@{}@/^1pc/[ru]^<{\{t_2\}}& & 3 \ar[lu]_{w} \ar@{}[ll]^<{\{t_1,t_2\}} \ar@/^1pc/[r]^{w} \ar@/_0.7pc/[ll]_{w} & 4 \ar[l]^{w} \ar@{}[lu]_<{\{t_1,t_2\}}\\
}$
\caption{\label{fig:inefficient}An example of an inefficient network}
\end{figure}

Consider as the initial Nash equilibrium the joint strategy $s = (t_2,
t_2, t_1, t_1)$. It is easy to check that this is indeed a Nash
equilibrium, with the payoff equal to $w - \theta$ for all nodes.
Suppose now that we remove product $t_1$ from the product set of node
$3$. We claim that the unique improvement path then leads to the Nash
equilibrium in which all nodes adopt $t_2$.

To see this, note that node $3$ moves first in any improvement path
and it has a unique choice, $t_2$. Then node $4$ moves and necessarily
switches to $t_2$. This yields a Nash equilibrium in which each node
selected $t_2$ with the payoff of $2w - \theta$, which
is strictly better than the payoff in $s$.  \HB
\end{example}

\subsection{Unsafe networks}

Finally, we analyze the following notion.  We call a social network
$\snet$ \bfe{unsafe} if $\mathcal{G}(\snet)$ has a Nash equilibrium,
while for some contraction $\snet'$ of $\snet$, $\mathcal{G}(\snet')$
does not. The following example shows that unsafe networks exist.

\begin{example}
\label{ex:unsafe}
\normalfont Let $\snet_1$ be the modification of the network $\snet$
given in Figure~\ref{fig:noNe-cycle} where node 1 has the product set
$\{t_1,t_2,t_4\}$, where $\theta(1,t_4) <r_2$. Then the joint strategy
$(t_4,t_3,t_3)$ is a Nash equilibrium in $\mathcal{G}(\snet_1)$. Now
consider the contraction $\snet_1'$ of $\snet_1$ where product
$t_4$ is removed from node 1. Then $\snet_1'$ is the network $\snet$,
which as we saw in Example~\ref{ex:noNe1} has no Nash equilibrium.
\HB
\end{example}

\section{Conclusions}

In this paper we studied dynamic aspects of social networks with
multiple products using the basic concepts of game theory.  We used
the model of social networks, originally introduced in \cite{AM11}
that we subsequently studied using game theory in \cite{SA12}, \cite{SA13}
and \cite{AMS13}.

However, in contrast to these three references the product adoption in
this paper is obligatory. This led to some differences. For example,
in contrast to the case of \cite{SA12}, a Nash equilibrium does not
need to exist when the underlying graph is a
simple cycle.  Further, in contrast to the setup of \cite{AMS13}, we
were able to construct a social network that exhibits the strongest
form of the paradox of choice.  On the other hand, some complexity
results, namely the ones concerning weakly acyclic games, remain the
same as in \cite{SA12}, though the proofs had to be appropriately
modified.

\bibliographystyle{eptcs}
\bibliography{e.bib}

\end{document}